%% file: online_thermal.tex
\documentclass[11pt]{article}
\usepackage[latin1]{inputenc}
\usepackage{a4wide,epsfig,amsmath,amssymb,a4wide}

\usepackage{color} % for pstex exported from xfig
\input{macros.tex}

\date{}

\begin{document}

\title{Algorithms for Temperature-Aware Task \\
	Scheduling in Microprocessor Systems}

\author{Marek Chrobak%
\thanks{Department of Computer Science,
        University of California, Riverside, CA 92521, USA.
        Supported by NSF grants OISE-0340752 and CCF-0729071.}
\and
Christoph D\"urr%
\thanks{CNRS, LIX UMR 7161, Ecole Polytechnique 
        91128 Palaiseau, France.
         Supported by ANR Alpage.}
\and
Mathilde Hurand%
\footnotemark[2]
\and
Julien Robert%
\thanks{Laboratoire de l'Informatique du Parall\'elisme,
Ecole Normale Sup\'erieure de Lyon; ENS Lyon, France.}
}

%%%%%%%%%%%%%%%%%%%%%%%%%%%%%%%%%%%%%%%%%%%%%%%%%%%%%%%%%%%%%%%%%%%%%%%%%%%%%%
%%%%%%%%%%%%%%%%%%%%%%%%%%%%%%%%%%%%%%%%%%%%%%%%%%%%%%%%%%%%%%%%%%%%%%%%%%%%%%
%%%%%%%%%%%%%%%%%%%%%%%%%%%%%%%%%%%%%%%%%%%%%%%%%%%%%%%%%%%%%%%%%%%%%%%%%%%%%%

\maketitle

\begin{abstract}
We study scheduling problems motivated
by recently developed techniques for microprocessor
thermal management at the operating systems level. 
The general scenario can be described as follows.
The microprocessor's temperature is controlled by
the hardware thermal management system that continuously
monitors the chip temperature and automatically reduces the
processor's speed as soon as the thermal threshold
is exceeded. Some tasks are more CPU-intensive than other
and thus generate more heat during execution. The
cooling system operates non-stop, reducing (at an
exponential rate) the deviation of the
processor's temperature from the ambient temperature.
As a result, the processor's temperature, and thus
the performance as well, depends on the
order of the task execution. Given a variety of possible
underlying architectures, models for 
cooling and for hardware thermal management, as well
as types of tasks, this scenario gives
rise to a plethora of interesting and never studied
scheduling problems. 

We focus on scheduling real-time jobs in a simplified
model for cooling and thermal management.
A collection of unit-length jobs is given, each
job specified by its release time, deadline and heat contribution. 
If, at some time step, 
the temperature of the system is $\tmp$ and the processor executes
a job with heat contribution $h$, then the temperature
at the next step is $(\tmp+h)/2$. The temperature cannot
exceed the given thermal threshold $\tmpthreshold$.
The objective is to maximize the throughput, that is, 
the number of tasks that meet their deadlines. We prove that, in the 
offline case, computing the optimum schedule is {\NP}-hard, even if
all jobs are released at the same time. In the online case,
we show a $2$-competitive deterministic algorithm and a matching lower bound. 
\end{abstract} 

%%%%%%%%%%%%%%%%%%%%%%%%%%%%%%%%%%%%%%%%%%%%%%%%%%%%%%%%%%%%%%%%%%%%%%%%%%%%%%
%%%%%%%%%%%%%%%%%%%%%%%%%%%%%%%%%%%%%%%%%%%%%%%%%%%%%%%%%%%%%%%%%%%%%%%%%%%%%%
%%%%%%%%%%%%%%%%%%%%%%%%%%%%%%%%%%%%%%%%%%%%%%%%%%%%%%%%%%%%%%%%%%%%%%%%%%%%%%

\section{Introduction}

%%%%%%%%%%%%%%%%%%%%%%%%%%%%%%%%%%%%%%%%%%%%%%%%%%%%%%%%%%%%%%%%%%%%%%%%%%%%%%

\paragraph{Background.}
The problem of managing the temperature of processor systems is
not new; in fact, the system builders had to deal with this challenge
since the inception of computers. Since early 1990s, the
introduction of battery-operated laptop computers and sensor systems
highlighted the related issue of controlling the energy consumption.

Most of the initial work on these problems was hardware and
systems oriented, and only during the last decade substantial
progress has been achieved on developing models and
algorithmic techniques for microprocessor temperature and energy management.
This work proceeded in several directions. One direction is based
on the fact that the energy consumption is a fast growing
function of the processor speed (or frequency).
Thus we can save energy by simply slowing down the processor. 
Here, algorithmic research focussed on \emph{speed scaling} --
namely dynamically adjusting the processor speed over time to optimize
the energy consumption while ensuring that the system meets
the desired performance requirements.
Another technique (applicable to the whole system, not
just the microprocessor) involves \emph{power-down strategies}, 
where the system is powered-down or even completely
turned off when some of its components are idle. 
Since changing the power level of a 
system introduces some overhead, scheduling the work to minimize
the overall energy usage in this model becomes a challenging
optimization problem. 

Models have also been
developed for the processor's thermal behavior. Here, the
main objective is to ensure that the system's temperature
does not exceed the so-called \emph{thermal threshold}, above
which the processor cannot operate correctly, or may even
be damaged. In this direction, techniques and algorithms have been
proposed for using speed-scaling to optimize the system's 
performance while maintaining the temperature below the
threshold.

We refer the reader to the survey by Irani and Pruhs \cite{IraPru05},
and references therein,
for more in-depth information on the models and algorithms
for thermal and energy management.

%%%%%%%%%%%%%%%%%%%%%%%%%%%%%%%%%%%%%%%%%%%%%%%%%%%%%%%%%%%%%%%%%%%%%%%%%%%%%%

\paragraph{Temperature-aware scheduling.}
The above models address energy and thermal 
management at the micro-architecture level. In contrast, 
the problem we study in this paper addresses the issue
of thermal management at the operating systems level.
Most of the previous work in this direction focussed
on multi-core systems, where one can
move tasks between the processors to minimize the maximum temperature
\cite{MerBel06,BeWeWK03,CCFHWB07,DonMar06,KuSPJ06,KuChBB06,GoPoVi04,MoChRS05}.
However, as it has been recently discovered, even in 
single-core systems one can exploit variations in heat contributions
among different tasks to reduce the processor's temperature
through appropriate task scheduling
\cite{BeWeWK03,GoPoVi04,DonMar06,KuSPJ06,YaXiCZ08}.
In this scenario, the microprocessor's temperature is controlled by
the hardware dynamic thermal management (DTM)
system that continuously
monitors the chip temperature and automatically reduces the
processor's speed as soon as the thermal threshold
(maximum safe operating temperature) is exceeded. 
Typically, the frequency is reduced by half, although it can be
further reduced to one fourth or even one eighth, if needed. 
Running at a lower frequency, the CPU generates less heat.
The cooling system operates non-stop, reducing (at an
exponential rate) the deviation of the
processor's temperature from the ambient temperature.
Once the chip cools down to below the
threshold, the frequency is increased again. 

Different tasks use different microprocessor units in 
different ways; in particular, some tasks are more CPU-intensive than other. 
As a result, the processor's thermal behavior -- and thus
the performance as well -- depends on the order of the task execution. 
In particular, Yang {\etal}~\cite{YaXiCZ08} point out that, based on
the standard model for the microprocessor thermal behavior,
for any given two tasks, scheduling the ``hotter'' job before 
the ``cooler'' one, results in a lower 
final temperature than after the reverse order.
They take advantage of this phenomenon to reduce the number
of DTM invocations, thus improving the performance of the OS scheduler.

With multitudes of possible
underlying architectures (for example, single- vs. multi-core systems), 
models for cooling and hardware thermal management, as well
as types of jobs (real-time, batch, etc.), the scenario
outlined above gives
rise to a plethora of interesting and never yet
studied scheduling problems. 

%%%%%%%%%%%%%%%%%%%%%%%%%%%%%%%%%%%%%%%%%%%%%%%%%%%%%%%%%%%%%%%%%%%%%%%%%%%%%%

\paragraph{Our model.}
We focus on scheduling real-time jobs in a somewhat
simplified model for cooling and thermal management. 
The time is divided into unit time slots and each job
has unit length. (These jobs represent unit slices of
the processes present in the OS scheduler's queue.)
We assume that the heat contributions of these jobs
are known. This is counterintuitive, but reasonably
realistic, for, as discussed in \cite{YaXiCZ08},
these values can be well approximated using
appropriate prediction methods.

In our thermal model we assume, without loss of generality,
that the ambient temperature is $0$ and that the
heat contributions are expressed in the units of
temperature (that is, by how much they would increase the
chip temperature in the absence of cooling). 
In reality \cite{YaXiCZ08}, during the execution of a
job, its heat contribution is spread over the whole time
slot and so is the effect of cooling; thus, the final
temperature can be expressed using an integral function.
In this paper, we use a simplified model where
we first take into account the job's heat contribution,
and then apply the cooling, where the cooling simply
reduces the temperature by half.

Finally, we assume that only one processor frequency
is available. Consequently, if there is no job whose
execution does not cause a thermal violation, the
processor must stay idle through the next time slot.

%%%%%%%%%%%%%%%%%%%%%%%%%%%%%%%%%%%%%%%%%%%%%%%%%%%%%%%%%%%%%%%%%%%%%%%%%%%%%%

\paragraph{Our results.}
Summarizing, our scheduling problem can be now formalized as follows.
A collection of unit-length jobs is given, each
job $j$ with a release time $r_j$, deadline $d_j$ and heat 
contribution $h_j$. If, at some time step, the
temperature of the system is $\tmp$ and the processor executes
a job $j$, then the temperature at the next step is $(\tmp+h_j)/2$.
The temperature cannot exceed the given thermal threshold $\tmpthreshold$.
The objective is to compute a schedule which
maximizes the number of tasks that meet their deadlines. 

We prove that in the offline case computing
the optimum schedule is {\NP}-hard, even if all jobs are released
at the same time and have equal deadlines. 
In the online case, we show a $2$-competitive deterministic
algorithm and a matching lower bound.

%%%%%%%%%%%%%%%%%%%%%%%%%%%%%%%%%%%%%%%%%%%%%%%%%%%%%%%%%%%%%%%%%%%%%%%%%%
%%%%%%%%%%%%%%%%%%%%%%%%%%%%%%%%%%%%%%%%%%%%%%%%%%%%%%%%%%%%%%%%%%%%%%%%%%
%%%%%%%%%%%%%%%%%%%%%%%%%%%%%%%%%%%%%%%%%%%%%%%%%%%%%%%%%%%%%%%%%%%%%%%%%%

\section{Terminology and Notation}
\label{sec: Terminology and Notation}

The input consists of $n$ unit-length jobs that we number $1,2,\ldots,n$. 
Each job $j$ is specified by a triple $(r_j,d_j,h_j)\in \mathbb N \times \mathbb N \times \mathbb Q$, where
$r_j$ is its release time, $d_j$ is the deadline and $h_j$ is its heat contribution.
The time is divided into unit-length slots and each job can be executed in any time
slot in the interval $[r_j,d_j]$. By $\tmp_u$ we denote the processor
temperature at time $u$. The initial temperature is $\tmp_0 = 0$, and
it changes according to the following rules: if the
temperature of the system at a time $u$ is $\tmp_u$ and the processor executes
a job $j$ then the temperature at time $u+1$ is 
$\tmp_{u+1} = (\tmp_u+h_j)/2$. The temperature
cannot exceed the given thermal threshold $\tmpthreshold$.
Without loss of generality, we assume that $\tmpthreshold=1$.
Thus if $(\tmp_u+h_j)/2 > 1$ then $j$ cannot be executed at time $u$.
Idle slots are treated as executing a job with heat contribution $0$,
that is, after an idle slot the temperature decreases by half.

Given an instance, as above, the objective is to compute a schedule
with maximum \emph{throughput}, where throughput is defined as the
number of completed jobs. Extending the standard notation for
scheduling problems, we denote the offline version of this problem by 
$1|r_{i},h_{i},p_i=1|\sum U_{i}$. 

In the online version, denoted $1|\textrm{online-}r_{i},h_{i},p_i=1|\sum U_{i}$,
jobs are available to the algorithm at their release time. Scheduling
decisions of the algorithm cannot depend on the jobs that have not
yet been released.

%%%%%%%%%

\paragraph{Example.}
Suppose we have four jobs, specified in notation $j \to (r_j,d_j,h_j)$:
$1\to (0,2,0.4)$, $2\to (0,4,0.6)$, $3\to (2,3,1.9)$, $4\to (4,6,0.8)$.

\begin{figure}[ht]
\begin{center}
\includegraphics[width=2.5in]{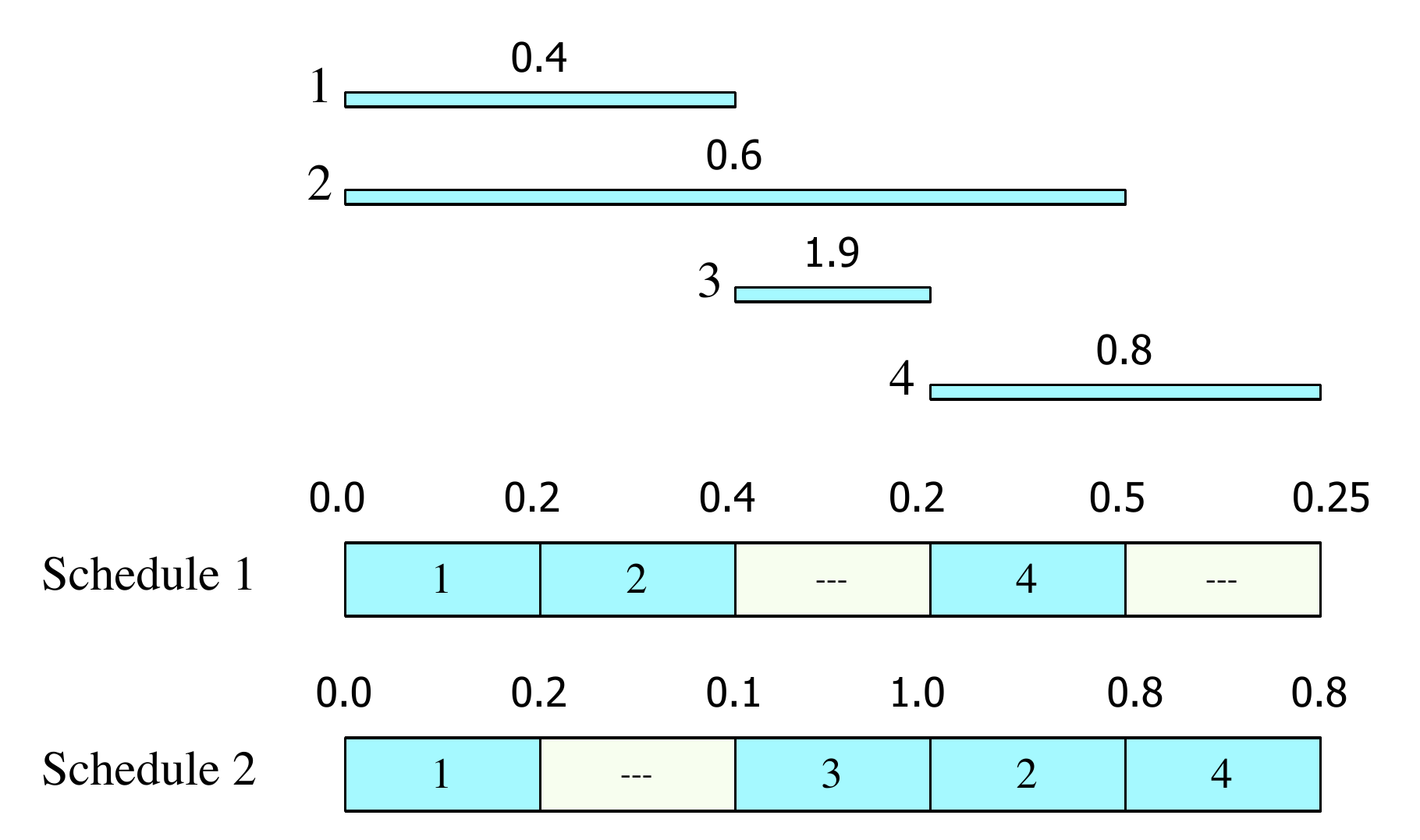}
\caption{Example of two schedules.}
\label{fig: example 1}
\end{center}
\end{figure}

Figure~\ref{fig: example 1} shows these jobs and two schedules.
Numbers above the schedules denote temperatures. In the first
schedule, when we schedule job $2$ at time $1$, the
processor is too hot to execute job $3$, so it will not
complete job $3$ at all. In the second schedule, we stay
idle in step 2, allowing us to complete all jobs.

%%%%%%%%%%%%%%%%%%%%%%%%%%%%%%%%%%%%%%%%%%%%%%%%%%%%%%%%%%%%%%%%%%%%%%%%%%
%%%%%%%%%%%%%%%%%%%%%%%%%%%%%%%%%%%%%%%%%%%%%%%%%%%%%%%%%%%%%%%%%%%%%%%%%%
%%%%%%%%%%%%%%%%%%%%%%%%%%%%%%%%%%%%%%%%%%%%%%%%%%%%%%%%%%%%%%%%%%%%%%%%%%

\section{The NP-Completeness Proofs}
\label{sec: The NP-Completeness Proof}

In this section we prove that the scheduling problem 
$1|r_{i},h_{i},p_i=1|\sum U_{i}$
is {\NP}-hard. For the sake of exposition, we start with a proof for the
general case, and later we give a proof for the special case when all
release times and deadlines are equal.

\begin{theorem}
The offline problem $1|r_{i},h_{i},p_i=1|\sum U_{i}$ is {\NP}-hard.
\end{theorem}

\begin{proof}
We use a reduction from the {\ThreePartition} problem, defined as follows:
we are given a set $S$ of $3n$ positive integers $a_{1},\ldots, a_{3n}$ 
such that $\beta/4<a_{i}<\beta/2$ for all $i$, where $\beta=\frac1n \sum a_{i}$. The goal is to 
partition $S$ into $n$ subsets, each subset with the total sum equal exactly $\beta$.
(By the assumption on the $a_i$, each subset will have to have exactly $3$ elements.)
This partition of $S$ will be called a \emph{$3$-partition}.
{\ThreePartition} is well-known to be {\NP}-hard \cite{GarJoh79} in the strong sense,
that is, even if $\max_i a_i \le p(n)$, for some polynomial $p(n)$.

We now describe the reduction. For the given instance of {\ThreePartition}
we construct an instance of $1|r_{i},h_{i},p_i=1|\sum U_{i}$ with
$4n$ jobs. These jobs will be of two types:
\begin{description}
	\item{(i)}
	First, we have $3n$ jobs that correspond to the instance of {\ThreePartition}.
	For every $1\le i \le 3n$ we create a job of heat contribution $2-2^{1-a_{i}}$, 
	release time $1$ and deadline $n(\beta+1)$.
	\item{(ii)}
	Next, we create $n$ additional ``gadget" jobs. These jobs are tight,
	meaning that their deadline is just one unit after the release time. 
	The first of these jobs has heat contribution $2$ and release time $0$. 
	Then, for each $1\le j \le n-1$, we have a job with heat contribution $1$ and 
	release time $j(\beta+1)$.
\end{description}
We claim that $S$ has a 3-partition if and only if the
instance of $1|r_{i},h_{i},p_i=1|\sum U_{i}$ constructed above
has a schedule with throughput $4n$ (that is, with all jobs completed).

The main idea is this: Imagine that at some moment the 
temperature is $1$, and we want to schedule a job of heat contribution $2-2^{1-x}$, for some integer $x\ge 1$. 
Then we must first wait $x-1$ time units, so that the processor cools down to
$(\frac12)^{x-1}=2^{1-x}$, before we can schedule that job, and then right at the 
completion of this job the temperature is $1$ again. 
The analogous property holds, in fact, even if 
at the beginning the temperature was some $\tmp > 1/2$,
except that then, after completing this job, the new
temperature will be 
$\tmp' = (\tmp 2^{1-x} + 2 - 2^{1-x})/2 
		> (2^{1-x}/2 + 2 - 2^{1-x})/2 
		= 1 - 2^{-x-1} 
		\ge 1/2$,
that is, it may be different than $\tmp$ but still strictly greater than $1/2$.
With this observation in mind, the proof of the above claim is quite easy.

$(\Leftarrow)$
First we show that if there is a solution to the scheduling problem, then $S$ has
a 3-partition. Note that the tight jobs divide the time 
into $n$ intervals of length $\beta$ each. Also each of the $3n$ other jobs is scheduled 
in exactly one of these intervals. This defines a partition of $S$ into
$n$ sets. 

Now we claim that after every job execution the temperature is strictly more than $1/2$. 
This is true for the first job to be scheduled, since it has heat contribution 
$2$. Each other job in the instance, including the tight jobs, has heat 
contribution at least $1$. Therefore right after its execution the temperature is at 
least $1/2$, already if we take only this job into account.  But there is also a declining but 
non-zero temperature contribution from the first tight job. So overall the temperature 
after every execution is strictly more than $1/2$.  
 
Together with the earlier observation, this implies that every non-tight job of heat 
contribution $2^{1-{a_i}}$ must be preceded by $a_i-1$ idle units, thus using 
$a_i$ time slots in total. Therefore every set in the 
above mentioned partition has the total sum at most $\beta$. Since 
there are at most $n$ sets in the partition, the total sum of each must be 
exactly $\beta$. This proves that $S$ has a $3$-partition.

$(\Rightarrow)$
Now we show the other implication, namely that if $S$ has a
$3$-partition then there is a solution to the scheduling instance. Simply schedule the 
tight jobs at their release time. This divides the time into $n$ intervals of length $\beta$ 
each. 
Assign each of the $n$ sets in the partition to a distinct interval and schedule its jobs 
in this interval: every integer $a_i$ in the set corresponds to a job of heat contribution 
$2-2^{1-{a_i}}$, and we schedule it preceded with $a_i-1$ idle time units.  The jobs of the set 
can be scheduled in an arbitrary order, the important property being that, since their 
total sum is $\beta$, they all fit exactly in this interval. After all jobs in one
set are executed the temperature is exactly $1$, and during the execution the
temperature does not exceed $1$ (because we pad the schedule with enough idle slots).
All release time and deadline constraints are 
satisfied, so the scheduling instance has a feasible solution as well.

It remains to show that the above
instance of $1|r_{i},h_{i},p_i=1|\sum U_{i}$ can be produced in polynomial 
time from the instance of {\ThreePartition}. Indeed, every number $a_{i}$ is mapped 
to some number $2-2^{1-a_{i}}$, which is described with $O(a_{i})$ bits.  Since we 
assumed that all numbers $a_{i}$ for $1\le i\le n$ are polynomial in $n$, 
the reduction will take polynomial time, and the proof is complete.
\end{proof}

%%%%%%%%%%%%%

The above construction used the constraints of the release times and deadlines to
fix tight jobs that force a partition of the time into intervals. We can actually
prove a stronger result, namely that the problem remains {\NP}-complete even if all 
release times are $0$ and all deadlines are equal.
Why is it interesting? One common approach in designing on-line algorithms for
scheduling is to compute the optimal schedule for the pending jobs, and use this schedule
to make the decision as to what execute next. The set of pending jobs forms an instance
where all jobs have the same release time. 
Our {\NP}-hardness result does not necessarily imply that the
above method cannot work (assuming that we do not care about the running time of the online
algorithm), but it makes this approach much less appealing, since reasoning about the
properties of the pending jobs is likely to be very difficult.

%%%%%%%%%%

\begin{theorem}\label{thm: np-complete 2}
The offline problem $1|r_{i},h_{i},p_i=1|\sum U_{i}$ is strongly {\NP}-hard even for
the special case when jobs are released at time $0$ and
all deadlines are equal.
\end{theorem}

\begin{proof}
The reduction is from {\NumThreeDM}.
In this problem, we are given 3 sets $A, B, C$ of $n$ non-negative 
integers each and a positive integer $\beta$.
A \emph{3-dimensional numerical matching} is a set of $n$ triples 
$(a, b, c) \in A  \times B \times C$ such that each number is matched 
exactly once (appears in exactly one triple) and
all triples $(a,b,c)$ in the set satisfy $a + b + c=\beta$. 
{\NumThreeDM} is known to be {\NP}-complete even when
the values of all numbers are bounded by a 
polynomial in $n$, it is referenced as problem [SP16] in \cite{GarJoh79}. 
(Clearly, this problem is quite similar to {\ThreePartition} that we
used in the previous proof.)
  
Without loss of generality, we can assume (A1) that every $x\in A\cup B\cup C$ 
satisfies $x\le \beta$ and (A2) that $\sum_{x\in A\cup B\cup C}x = \beta n$.

We now describe the reduction. Let be the constant $\alpha=1/25$ and the function $f : x\mapsto  \alpha(1 + x/8\beta)$.
The instance of $1|r_{i},h_{i},p_i=1|\sum U_{i}$
will have  $4n+1$ jobs, all with release time $0$ and deadline $4n+1$.  
These jobs will be of two types:

\begin{description}
\item{(i)} First we have $3n$ jobs that correspond to the instance of 
{\NumThreeDM}:
for every $a\in A$, there is a job of heat contribution $8f(a)$,
for every $b\in B$, there is a job of heat contribution $4f(b)$, and
for every $c\in C$, there is a job of heat contribution $2f(c)$.
We call these jobs, respectively, \emph{$A$-jobs}, \emph{$B$-jobs} 
and \emph{$C$-jobs}.

\item{(ii)}
Next, we have $n+1$ ``gadget'' jobs. The first of these jobs has heat 
contribution $2$, and the remaining ones $1.75$. We call these jobs,
respectively, \emph{$2$-} and \emph{$1.75$-jobs}.
\end{description}

We claim that the instance $A,B,C,\beta$ has a numerical
$3$-dimensional matching if
and only if the instance of $1|r_{i},h_{i},p_i=1|\sum U_{i}$ that
we constructed has a schedule with all jobs completed not later than
at time $4n+1$.

The idea of the proof is that the gadget jobs are so hot that they
need to be scheduled only every 4-th time unit, separating the time into 
$n$ blocks of 3 consecutive time slots each. Every remaining job has 
a heat contribution that consists of two parts: a constant part 
($8\alpha$, $4\alpha$ or $2\alpha$) and a tiny 
variable part that depends on the instance of the matching problem. 
This constant part is so large that in every block there is a single 
$A$-job, a single $B$-job and a single $C$-job,
and they must be scheduled in that order. This defines a 
partitioning of $A,B,C$ into triplets of the form $(a,b,c)\in A\times B\times C$. 
Since the gadget jobs are so hot, 
they force every triple $(a,b,c)$ to satisfy $a+b+c\le \beta$. 
We now make this argument formal.

$(\Rightarrow)$
Suppose there is a solution to the instance of {\NumThreeDM}. 
We construct a schedule where all jobs complete at or before time $4n+1$.
Schedule the $2$-job at time $0$, and all
other gadget jobs every $4$-th time slot. Now the remaining slots are 
grouped into blocks consisting of $3$ consecutive time slots each.
For $i=1,2,\ldots,n$,
associate the $i$-th triple $(a, b, c)$ from the matching with 
the $i$-th block, and the corresponding $A$-,$B$- and $C$-jobs are 
executed in this block in the order $A,B,C$ --- see Figure~\ref{fig: abc schedules}.

\begin{figure}[ht]
\begin{center}
\includegraphics[width=5in]{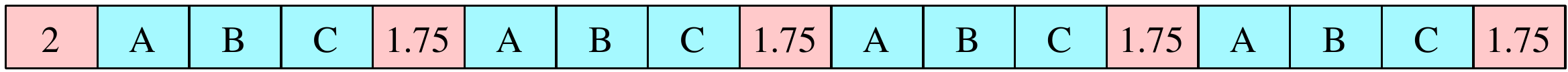}
\caption{The structure of the schedules in the proof.}
\label{fig: abc schedules}
\end{center}
\end{figure}

By construction, every job meets the deadline, so it remains to show that 
the temperature never exceeds $1$. The non-gadget jobs have all heat 
contribution smaller than $1$, by assumption (A1), so execution of a
non-gadget job cannot increase the temperature to above $1$, as long
as the temperature before was not greater than $1$.

Now we show by induction that right after the execution of a gadget job
the temperature is exactly $1$. This is clearly the case after execution 
of the first job, since its heat contribution is $2$. Now let $u$ be the 
time when a $1.75$-job is scheduled, and suppose that at time $u-3$ the 
temperature was $1$. Let $(a,b,c)$ be the triple associated with the block 
consisting of time slots between $u-3$ and $u$. Then, by $a+b+c=\beta$,  
at time $u$ the temperature is
\begin{eqnarray*}
  	\frac18 + \frac{8f(a)}8+\frac{4f(b)}4+\frac{2f(c)}2 
		\;=\; \frac18 + \alpha\left(3+\frac{a+b+c}{8\beta} \right)
		\;=\; \frac18+\alpha\left(\frac{24}{8}+\frac{1}{8}\right) 
		\;=\; \frac14.
\end{eqnarray*}
This shows that at time $u+1$, after scheduling a $1.75$-job, the temperature 
is again $(1.75+1/4)/2 = 1$. We conclude that the schedule is feasible.

$(\Leftarrow)$
Now we show the remaining implication.
Suppose the instance of  $1|r_{i},h_{i},p_i=1|\sum U_{i}$ constructed
above has a schedule where all jobs meet the deadline $4n+1$. We show that 
there exists a matching of $A,B,C$.
 
We first show that this schedule must have the form from 
Figure~\ref{fig: abc schedules}. First, note that
since all $4n+1$ jobs have deadline $4n+1$, all jobs must be scheduled 
without any idle time between time $0$ and $4n+1$. This means that the gadget 
job of heat contribution $2$ must be scheduled at time $0$,  because that is 
the only moment of temperature $0$. Also note that every job has heat contribution 
at least $2f(0)=2\alpha$.
  
Now we claim that all $1.75$-jobs have to be scheduled every $4$-th time unit.
This holds because two units after scheduling a $1.75$-job, the temperature is 
at least
\[
	\frac{1.75}{8} + \frac{2\alpha}{4} + \frac{2\alpha}{2} > 1/4.
\]
Therefore two executions of $1.75$-jobs must be at least $4$ time units apart,
and this is only possible if they are scheduled exactly at times $4i$ for 
$i=1,\ldots,n$. 

We claim that after every execution of a gadget job, the temperature is at 
least $\tau=364/375$. Clearly this is the case after the execution of the $2$-job.
Now assume that at time $4i+1$, for some $i=0,\ldots,n-1$, the claim holds. 
Then at time $4i+5$, after the execution of the next $1.75$-job, the temperature 
is at least
\[
	\frac{\tau}{16} + \frac{2\alpha}{16} + \frac{2\alpha}{8} + \frac{2\alpha}{4} 
	+\frac{1.75}{2} = \tau.
\]
We show now that every block contains exactly one $A$-job, one $B$-job and 
one $C$-job, in that order. Towards contradiction, suppose that some 
$A$-job is scheduled at the 2nd position of a block, say at time $4i+2$ for some $i\in\{0,\ldots,n-1\}$. Its heat contribution is at least $8f(0)$.  
Therefore the temperature at time $4i+4$ would be at least
\[
  	\frac{\tau}{8}+\frac{2\alpha}{8} + \frac{8\alpha}{4} +\frac{2\alpha}{2} 
		> 1/4,
\]
contradicting that a $1.75$-job is scheduled at that time:
A similar argument shows that $A$-jobs cannot be scheduled at position $3$
in a block, and therefore the $1$st position of a block is always occupied 
by an $A$-job. 

By an analogous reasoning, we show that a $B$-job cannot be scheduled at the 
$3$rd position of some block. It it were scheduled there, 
the temperature at the end of block would be at least
 \[
  	\frac{\tau}{8}+\frac{8\alpha}{8} + \frac{2\alpha}{4} +\frac{4\alpha}{2} 
> 1/4,
  \]
again contradicting that a $1.75$-job is scheduled at the end of the block.

We showed that every block contains jobs that correspond to 
some triple $(a,b,c)\in A\times B\times C$. It remains to show that each
such triple satisfies $a+b+c=\beta$. Let $(a_{i},b_{i},c_{i})$ be the triple
corresponding to the $i$-th block, for $1\le i\le n$. 

Define $t_{0}=1$ and 
\begin{eqnarray*}
	t_{i} &=& \frac{1}{16}\cdot 8f(a_i) + \frac{1}{8}\cdot 4f(b_i)
					+ \frac{1}{4}\cdot 2f(c_i) + \frac{1}{2} \cdot 1.75
						\\
		  &=& \frac{1}{400} [ 374 + (a_i + b_i + c_i)/\beta ].
\end{eqnarray*}
Thus $t_i$ represents the contribution of the $i$th block and a following
gadget job to the temperature right after this gadget job. This implies that,
for $1\le k\le n$, the temperature at time $4k+1$ is exactly
$\sum_{i=0}^{k} (1/16)^{k-i }t_{i}$. By Assumption (A2),
$\sum_{i=1}^n (a_{i}+b_{i}+c_{i}) = n\beta$, and thus
$\sum_{i=1}^n t_i = \frac{15}{16}n$.

Define $p_{i}=t_{i}-15/16$ for $i=1,2,...,n$. From the previous
paragraph,
\begin{equation} 										\label{eq:pi0}
    \sum_{i=1}^{n} p_{i} \;=\; 0.
\end{equation}
As mentioned earlier, $\sum_{i=0}^{k} (1/16)^{k-i }t_{i}$ is 
the temperature at time $4k+1$, so we have
$\sum_{i=0}^{k} (1/16)^{k-i }t_{i} \le 1$. Therefore,
for all $1\le k \le n$ we get
\begin{eqnarray*}
 16^{-k} \sum_{i=1}^{k} 16^{i} p_{i} 
		&=& 16^{-k} \sum_{i=1}^{k} 16^i (t_i -15/16)
\\
 		&=& \sum_{i=0}^{k} (1/16)^{k-i} t_i 
				- (1/16)^k -15 \sum_{i=1}^{k} (1/16)^{k-i+1}	
\\
		&\le& 1 - (1/16)^k -15 (1  - (1/16)^k ) / 15 
		\;=\; 0.
\end{eqnarray*}
We conclude that for $k = 1,2,...,n$ we have
\begin{eqnarray}
	\sum_{i=1}^{k} 16^{i} p_{i} &\le& 0.
				\label{eq:pik}
\end{eqnarray}
To complete the proof, it remains to show that 
$p_i = 0$ for all $i$, for this will
imply that $a_{i}+b_{i}+c_{i}=\beta$, which in turn implies that there is
a matching. We prove this claim by contradiction. Suppose 
that not all $p_{i}$'s are zero. Let $\ell$ be the smallest index 
such that  $p_{\ell}>0$ and
\begin{equation}										\label{eq:ell0}
       p_{1}+\ldots+p_{\ell}\ge 0.
\end{equation}
Clearly, $\ell\ge 2$.
By the minimality of $\ell$, for every $2\le k\le\ell-1$ we have
\[
	p_{1}+\ldots+p_{k-1} \le 0
	 	\text{ \:\:and\:\: } p_{k}+\ldots+p_{\ell} \ge 0.
\]
There are $\sigma_i > 0$, $i=1,...,\ell$ ,such that
$\sum_{i=1}^j \sigma_i = 16^j$. Then
\begin{eqnarray*}
\sum_{j=1}^\ell 16^j p_j \;=\; \sum_{j=1}^\ell \sum_{i=1}^j \sigma_i p_j
		\;=\; \sum_{i=1}^{\ell-1} \sigma_i \sum_{j=i}^\ell p_i + \sigma_\ell p_\ell
			\; >\; 0,
\end{eqnarray*}
because all terms are non-negative and $p_\ell > 0$. This 
contradicts (\ref{eq:pik}). 

It remains to show that the above
instance of $1|h_{i},p_i=1|\sum U_{i}$ can be produced in polynomial 
time from the instance of {\NumThreeDM}. Indeed, every number $x\in A\cup B\cup C$ 
is mapped to some fraction, where both the denominator and numerator are linear 
in $x$ and $\beta$. Therefore if we represent fractions by writing the denominator 
and numerator, and not as some rounded decimal expansion,  
the reduction will take polynomial time, and the proof is complete.
\end{proof}

%%%%%%%%%%

Theorem~\ref{thm: np-complete 2} implies that other variants of temperature-aware
scheduling are {\NP}-hard as well. Consider for example the
problem $1|h_{i},p_i=1|C_{\max}$, where the objective is to minimize
the \emph{makespan}, that is, the maximum completion time.
In the decision version of this problem we ask if all jobs can
be completed by some given deadline $C$ -- which is exactly
what we proved above to be {\NP}-hard.
It also gives us {\NP}-hardness of
$1|h_{i},p_i=1|\sum C_j$. To prove this, we can use the decision version
of this problem where we ask if there is a schedule for which the
total completion time is at most $n(n-1)/2$ (where $n$ is the number
of jobs).

%%%%%%%%%%%%%%%%%%%%%%%%%%%%%%%%%%%%%%%%%%%%%%%%%%%%%%%%%%%%%%%%%%%%%%%%%%
%%%%%%%%%%%%%%%%%%%%%%%%%%%%%%%%%%%%%%%%%%%%%%%%%%%%%%%%%%%%%%%%%%%%%%%%%%
%%%%%%%%%%%%%%%%%%%%%%%%%%%%%%%%%%%%%%%%%%%%%%%%%%%%%%%%%%%%%%%%%%%%%%%%%%
 
\section{An Online Competitive Algorithm}
\label{sec: An Online Competitive Algorithm}

In this section we show that there is a $2$-competitive algorithm
for  $1|\textrm{online-}r_{i},h_{i},p_i=1|\sum U_{i}$.
We will show, in fact, that a large class of deterministic algorithms is
$2$-competitive. 

Given a schedule, we will say that a job $j$ is \emph{pending} at
time $u$ if $j$ is released, not expired (that is, $r_j\le u< d_j$)
and $j$ has not been scheduled before $u$.
If the temperature at time $u$ is $\tmp_u$ and $j$ is pending, then
we call $j$ \emph{admissible} if $\tmp_u+h_j \le 2$, that is,
$j$ is not too hot to be executed.

We say that a job $j$ \emph{dominates} a job $k$ if $j$ is both not
hotter and has the same or smaller deadline than $k$, 
that is $h_j\le h_k$ and $d_j\le d_k$. 
If at least one of these inequalities is strict, then we say
that $j$ \emph{strictly dominates} $k$.
An online algorithm is called \emph{reasonable} if  at each step 
(i) it schedules a job whenever one is admissible (the 
\emph{non-waiting property}), and, if there is one,
(ii) it schedules an admissible job that is not strictly dominated by 
another pending 
job. The class of reasonable algorithms contains, for example, the 
following two natural algorithms:
\begin{description}
	\item{{\CoolestFirst}:} Always schedule a coolest admissible
	 	job (if there is any), breaking ties in favor of jobs
		with earlier deadlines.
	\item{{\EarliestDeadlineFirst}:} Always schedule an 
	    admissible job (if there is one) with the earliest deadline,
		breaking ties in favor of cooler jobs.
\end{description}
If two jobs have the same deadlines and heat contributions, both
algorithms give preference to one of them arbitrarily.

%%%%%%%%%%

\begin{theorem}\label{thm: reasonable 2-competitive}
Any reasonable algorithm for $1|\textrm{online-}r_{i},h_{i},p_i=1|\sum U_{i}$ 
is $2$-competitive.
\end{theorem}

\begin{proof}
Let $\algA$ be any reasonable algorithm.
We fix some instance, and we compare the schedules produced by $\algA$
and the adversary on this instance. The proof is based on a
charging scheme that maps jobs executed by the 
adversary to jobs executed by $\algA$ in such a way
that no job in $\algA$'s schedule gets more than two charges.

\begin{figure}[ht]
\begin{center}
\includegraphics[width=6in]{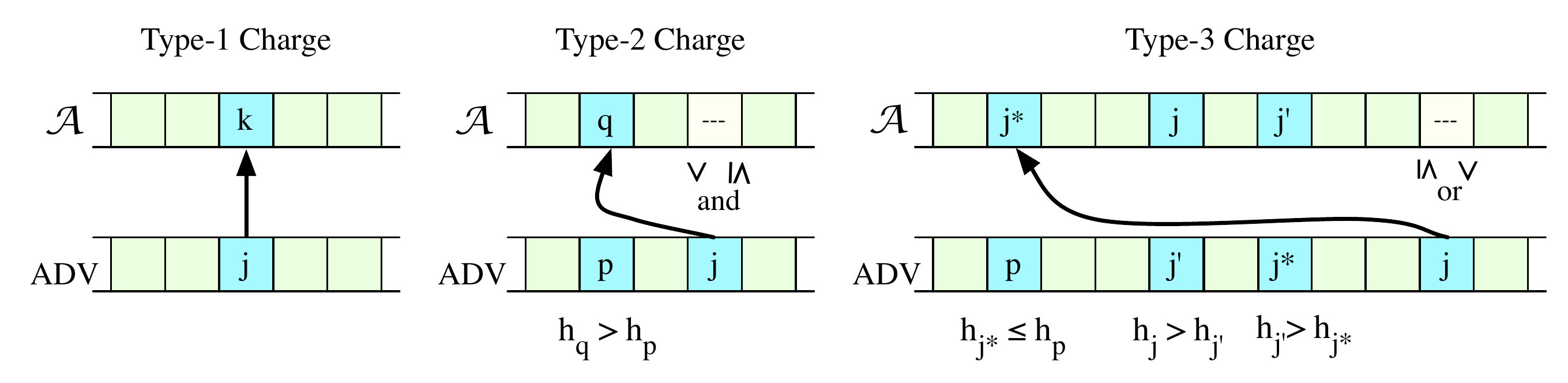}
  \caption{Four types of charges. The vertical inequality signs between
			the schedules show the relation between the temperatures.}
\label{fig: charging}
\end{center}
\end{figure}

We now describe this charging scheme. (See Figure~\ref{fig: charging} for
illustration.)
There will be three types of charges, depending on whether $\algA$
is busy or idle at a given time step and on the relative temperatures
of the schedules of $\algA$ and the adversary.
The temperature at time $u$ in the schedules of $\algA$
and the adversary will be denoted by $\tmp_u$ and $\tmp'_u$, respectively.

Suppose that at some time $u$, $\algA$ schedules a job $k$ while
the adversary schedules a job $j$, or that the adversary is idle
(we treat this case as if executing a job $j$ with $h_j = 0$.)
Then step $u$ will be called a \emph{relative-heating step} 
if $k$ is strictly hotter than $j$, that is $h_k > h_j$. Note
that if $\tau_v > \tau'_v$ for some time $v$, then a relative-heating
step must have occurred before time $v$.

Consider now a job $j$ scheduled by the adversary, say at time $v$.
The charge of $j$ is defined as follows:

\emph{Type~1 Charges:}
If $\algA$ schedules a job $k$ at time $v$, charge $j$ to $k$. 
Otherwise, we are idle, and then we have two more cases.

\emph{Type-2 Charges:}
Suppose that $\algA$ is hotter than the adversary 
at time $v$ but not at $v + 1$, that is $\tmp_u > \tmp'_u$ and
$\tmp_{u+1}\le \tmp'_{u+1}$.
In this case we charge $j$ to the job $q$ executed by $\algA$ in the last
relative-heating step before $v$. (As explained above, this step
exists.)

\emph{Type-3 Charges:}
Suppose now that either $\algA$ is not hotter than the adversary at $v$
or $\algA$ is hotter than the adversary at $v + 1$. In other words,
$\tmp_v\le \tmp'_v$ or $\tmp_{v+1} >\tmp'_{v+1}$.
(Note that in the latter case we must also have 
$\tmp_v > \tmp'_v$ as well, since the algorithm is idle.)

We claim that $\tmp_v + h_j \le 2$, which means that neither
$j$ or any job $\ell$ with $h_\ell \le h_j$ can be pending
at $v$. To justify this, we consider the two sub-cases of the
condition for type-3 charges.
If $\tmp_v \le \tmp'_v$, the claim is trivial, since
then $\tmp_v +h_j \le \tmp'_v + h_j \le 2$, because the 
adversary executes $j$ at time $v$.
So assume now that $\tmp_{v+1} >\tmp'_{v+1}$.
Since $\algA$ is idle, we have $\tmp_{v+1}\le 1/2$. Therefore 
$h_j = 2\tmp'_{v+1} -\tmp'_v 
    \le 2\tmp'_{v+1} 
    \le 1$, 
and the claim follows because $\tmp_v\le 1$ as well.

From the above claim, $j$ was scheduled by $\algA$ at some time $u < v$.
To find a job that we can charge $j$ to, we construct a chain of jobs
$j,j',j'',\ldots,j^\ast$ with strictly decreasing heat contributions.
Further, all jobs in this chain except $j^\ast$ will be executed
by $\algA$ at relative-heating steps. This chain will be determined 
uniquely by $j$, and we will charge $j$ to $j^\ast$.
If, at time $u$, the adversary is idle or schedules an equal or hotter job, 
then $j^\ast = j$, that is, we charge $j$ to itself (its ``copy" in 
$\algA$'s schedule). Otherwise, if the adversary schedules $j'$ at
time $u$ then $j'$ is strictly cooler than $j$, that is $h_{j'} < h_j$. 
Now we claim that the algorithm schedules $j'$ at some time before $v$. 
Indeed, if $j'$ is scheduled before $u$, we are done.
Otherwise, $j'$ is pending at $u$, and, since 
the algorithm never schedules a dominated job, we must have 
$d_{j'}\ge d_j \ge v+1$. By our earlier observation and
by $h_{j'} < h_j$, if $\algA$ did not schedule
$j'$ before $v$, then $j'$ would have been admissible at $v$, 
contradicting the fact $\algA$ is idle at $v$.
So now the chain is $j,j'$.
Let $u'$ be the time when $\algA$ schedules $j'$.
If the adversary is idle at time $u'$ or if $j'$ is not hotter than the
job executed by the adversary at time $u'$, we take $j^\ast = j'$.
Otherwise, we take $j''$ to be the
job executed by the adversary at time $u'$, and so on. 
This process must end at some point, since we deal
with strictly cooler and cooler jobs. So the job $j^\ast$
is well-defined.

\medskip

This completes the description of the charging scheme.
Now we show that any job scheduled by $\algA$ will get at most two
charges. Obviously, each job in $\algA$'s schedule gets at most
one type-1 charge. In-between
any two time steps that satisfy the condition of the type-2 charge
there must be a relative-heating step, so the type-2 charges are assigned
to distinct relative-heating steps. As for type-3
charges, every chain defining a type-3 charge is uniquely defined
by the first considered job, and these chains are disjoint.
Therefore type-3 charges are assigned to distinct jobs.

Now let $k$ be a job scheduled by $\algA$ at some time $v$. By the previous 
paragraph, $k$ can get at most one charge of each type. We claim that
$k$ cannot get charges of each type $1$, $2$ and $3$.
Indeed, if $k$ receives a type-1 charge, then the adversary is not idle at time $v$, 
and schedules some job $\ell$. 
If $k$ also receives a type-2 charge, then $v$ must be a relative-heating step,
that is $h_k > h_\ell$.
But to receive a type-3 charge, $k$ would be the last
job $j^\ast$ in a chain of some job $j$, and since the chain
was not extended further, we must have $h_k\le h_\ell$.
So type-1, type-2 and type-3 charges cannot coincide.

Summarizing the argument, we have that
every job scheduled by the adversary is charged to some 
job scheduled by $\algA$, and every job scheduled by $\algA$
receives no more than $2$ charges.
Therefore the competitive ratio of $\algA$ is not more than $2$.
\end{proof}

%%%%%%%%%%%%%%%%%%%%%%%%%%%%%%%%%%%%%%%%%%%%%%%%%%%%%%%%%%%%%%%%%%%%%%%%%%
%%%%%%%%%%%%%%%%%%%%%%%%%%%%%%%%%%%%%%%%%%%%%%%%%%%%%%%%%%%%%%%%%%%%%%%%%%
%%%%%%%%%%%%%%%%%%%%%%%%%%%%%%%%%%%%%%%%%%%%%%%%%%%%%%%%%%%%%%%%%%%%%%%%%%

\section{A Lower Bound on the Competitive Ratio}
\label{sec: A Lower Bound on the Competitive Ratio}

\begin{theorem}\label{thm: lower bound of 2}
Every deterministic online algorithm for 
$1|\textrm{online-}r_{i},h_{i},p_i=1|\sum U_{i}$ has competitive ratio at least $2$.
\end{theorem}

\begin{proof}
Fix some deterministic online algorithm $\algA$.
We (the adversary) release a job 
$1\to (0,3,1.2)$ (in notation $j\to (r_j,d_j,h_j)$). 
If $\algA$ schedules it at time $0$, we release at time $1$
a tight job $2\to(1,2,1.6)$ 
and schedule it followed by $1$. $\algA$'s schedule is too
hot at time $1$ to execute job $2$. If $\algA$ does not
schedule job $1$ at time $0$, then we schedule it at $0$ and release
at time $2$ (and
schedule) a tight job $3\to(2,3,1.6)$ at time $2$. In this case, $\algA$
cannot complete both jobs $1$ and $3$ without violating the thermal
threshold.
In both cases we schedule two jobs, while $\algA$ schedules only one,
completing the proof. (See Figure~\ref{fig:lower-2}.)
\end{proof}

\begin{figure}[ht]
\begin{center}
\includegraphics[width=3in]{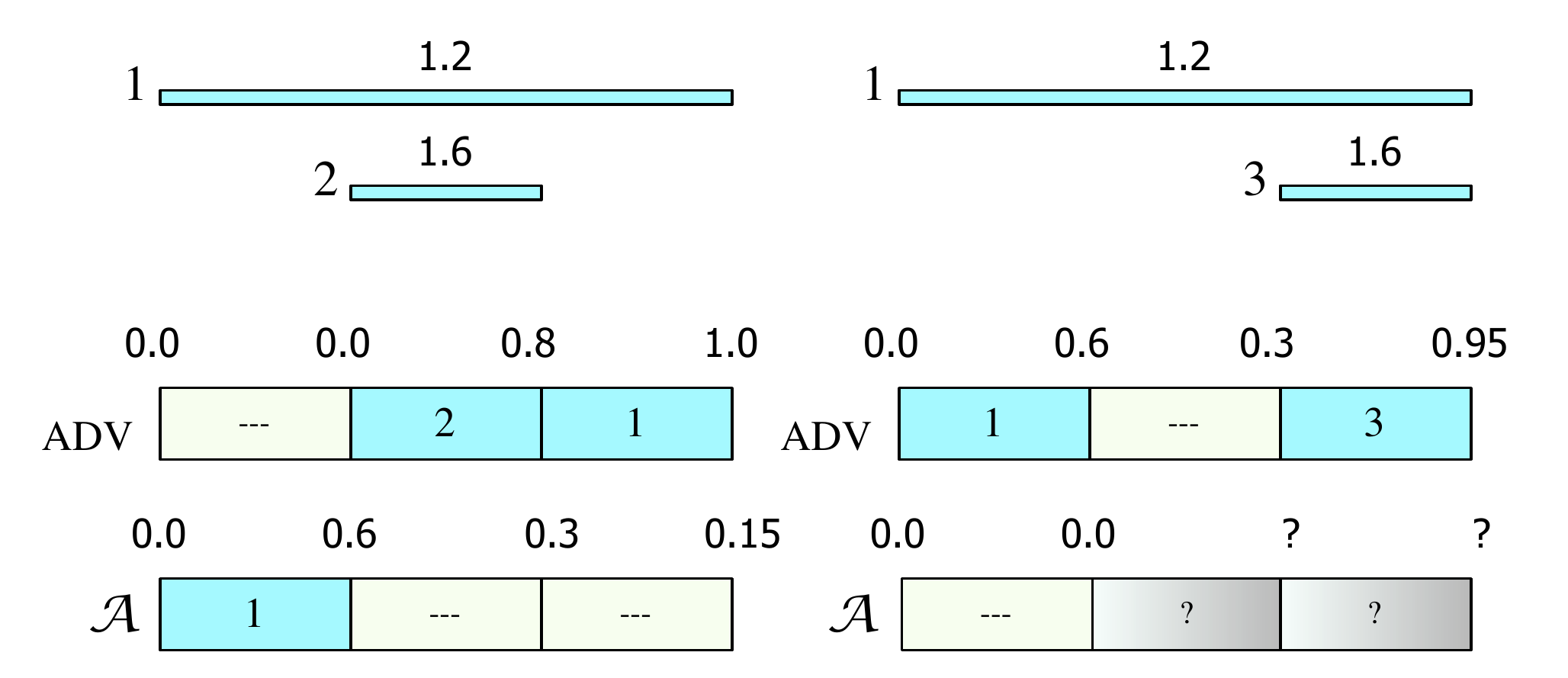}
  \caption{The lower bound for deterministic algorithms.}
\label{fig:lower-2}
\end{center}
\end{figure}

%%%%%%%%%%%%%%%%%%%%%%%%%%%%%%%%%%%%%%%%%%%%%%%%%%%%%%%%%%%%%%%%%%%%%%%%%%
%%%%%%%%%%%%%%%%%%%%%%%%%%%%%%%%%%%%%%%%%%%%%%%%%%%%%%%%%%%%%%%%%%%%%%%%%%
%%%%%%%%%%%%%%%%%%%%%%%%%%%%%%%%%%%%%%%%%%%%%%%%%%%%%%%%%%%%%%%%%%%%%%%%%%

\section{Final Comments}
\label{sec: Final Comments}

Many open problems remain. Perhaps the most intriguing one is to
determine the randomized competitive ratio for the problem we studied.
The proof of Theorem~\ref{thm: lower bound of 2} can easily 
be adapted to prove the lower bound of $1.5$, but we have 
not been able to improve the upper bound of $2$; this is, in fact,
the main focus of our current work on this scheduling problem.
One approach, based on
Theorem~\ref{thm: reasonable 2-competitive}, would be to
randomly choose, at the beginning of computation,
two different reasonable algorithms $\algA_1$, $\algA_2$,
each with probability $1/2$, and then deterministically execute
the chosen $\algA_i$. So far, we have been able to show that
for many natural choices for $\algA_1$ and $\algA_2$
(say, {\CoolestFirst} and {\EarliestDeadlineFirst}), this
approach does not work. 

Extensions of the cooling model can be considered,
where the temperature after executing $j$ is $(\tmp+h_j)/R$, for some
$R > 1$. Even this formula, however, is only a discrete 
approximation for the true model (see, for example, \cite{YaXiCZ08}),
and it would be interesting to see if the ideas behind our 
$2$-competitive algorithm can be adapted to these more realistic cases.

In reality, thermal violations do not cause the system to idle, but
only to reduce the frequency. With frequency reduced to half, a unit
job will execute for two time slots. Several frequency levels may
be available.

We assumed that the heat contributions are known. This is
counter-intuitive, but not unrealistic, since the "jobs"
in our model are unit slices of longer jobs. Prediction methods
are available that can quite accurately predict the heat
contribution of each slice based on the heat contributions of
the previous slices. Nevertheless, it may be interesting
to study a model where exact heat contributions are not known.

Other types of jobs may be studied. For real-time jobs, one can 
consider the case when not all jobs are equally important,
which can be modeled by assigning weights to jobs and
maximizing the weighted throughput. For batch jobs, other
objective functions can be optimized, for example the flow time.

One more realistic scenario would be to represent the whole
processes as jobs, rather then their slices. This naturally
leads to scheduling problems with preemption and with
jobs of arbitrary processing times. When the thermal
threshold is reached, the execution of a job is slowed down
by a factor of $2$. Here, a scheduling algorithm may decide to
preempt a job when another one is released or, say, when the
processor gets too hot.

Finally, in multi-core systems one can explore the migrations
(say, moving jobs from hotter to cooler cores) to keep the
temperature under control. This leads to even more scheduling
problems that may be worth to study.

%%%%%%%%%%%%%%%%%%%%%%%%%%%%%%%%%%%%%%%%%%%%%%%%%%%%%%%%%%%%%%%%%%%%%%%%%%
%%%%%%%%%%%%%%%%%%%%%%%%%%%%%%%%%%%%%%%%%%%%%%%%%%%%%%%%%%%%%%%%%%%%%%%%%%
%%%%%%%%%%%%%%%%%%%%%%%%%%%%%%%%%%%%%%%%%%%%%%%%%%%%%%%%%%%%%%%%%%%%%%%%%%

{
\bibliographystyle{plain}
\bibliography{temperature}
}

\end{document}

%% file: macros.tex
\newcommand{\etal}{{\it et~al.\ }}

\newtheorem{theorem}{Theorem}

\newenvironment{proof}{{\noindent\bf Proof:\/}}{\hfill$\Box$\vskip 0.1in}

\newcommand{\algA}{{\cal A}}

\newcommand{\tmp}{\tau}
\newcommand{\tmpthreshold}{T}

\newcommand{\NP}{\mbox{\sf NP}}
\newcommand{\ThreePartition}{{\mbox{\sc 3-Partition}}}
\newcommand{\NumThreeDM}{{\mbox{\sc Numerical-3D-Matching}}}

\newcommand{\EarliestDeadlineFirst}{\mbox{\sf EarliestDeadlineFirst}}
\newcommand{\CoolestFirst}{\mbox{\sf CoolestFirst}}

%% file: online_thermal.bbl
\begin{thebibliography}{10}

\bibitem{BeWeWK03}
F.~Bellosa, A.~Weissel, M.~Waitz, and S.~Kellner.
\newblock Event-driven energy accounting for dynamic thermal management.
\newblock In {\em Workshop on Compilers and Operating Systems for Low Power},
  2003.

\bibitem{CCFHWB07}
J.~Choi, C-Y. Cher, H.~Franke, H.~Hamann, A.~Weger, and P.~Bose.
\newblock Thermal-aware task scheduling at the system software level.
\newblock In {\em International Symposium on Low Power Electronics and
  Design,}, pages 213--218, 2007.

\bibitem{GarJoh79}
M.R. Garey and D.S. Johnson.
\newblock {\em Computers and Intractability: A Guide to the Theory of
  {NP}-Completeness}.
\newblock W.H.Freeman and Co., 1979.

\bibitem{GoPoVi04}
M.~Gomaa, M.~D. Powell, and T.~N. Vijaykumar.
\newblock Heat-and-run: leveraging smt and cmp to manage power density through
  the operating system.
\newblock {\em SIGPLAN Not.}, 39(11):260--270, 2004.

\bibitem{IraPru05}
S.~Irani and K.~R. Pruhs.
\newblock Algorithmic problems in power management.
\newblock {\em SIGACT News}, 36(2):63--76, 2005.

\bibitem{DonMar06}
M.~Martonosi J.~Donald.
\newblock Techniques for multicore thermal management: Classification and new
  exploration.
\newblock In {\em Proceedings of the International Symposium on Computer
  Architecture}, pages 78--88, 2006.

\bibitem{KuSPJ06}
A.~Kumar, L.~Shang, L-S. Peh, and N.~K. Jha.
\newblock {HybDTM}: a coordinated hardware-software approach for dynamic
  thermal management.
\newblock In {\em DAC '06: Proceedings of the 43rd Annual Conference on Design
  Automation}, pages 548--553, 2006.

\bibitem{KuChBB06}
E.~Kursun, C.-Y. Cher, A.~Buyuktosunoglu, and P.~Bose.
\newblock Investigating the effects of task scheduling on thermal behavior.
\newblock In {\em the 3rd Workshop on Temperature-Aware Computer Systems},
  2006.

\bibitem{MerBel06}
A.~Merkel and F.~Bellosa.
\newblock Balancing power consumption in multiprocessor systems.
\newblock {\em SIGOPS Oper. Syst. Rev.}, 40(4):403--414, 2006.

\bibitem{MoChRS05}
J.~Moore, J.~Chase, P.~Ranganathan, and R.~Sharma.
\newblock Making scheduling "cool": temperature-aware workload placement in
  data centers.
\newblock In {\em ATEC'05: Proceedings of the USENIX Annual Technical
  Conference 2005 on USENIX Annual Technical Conference}, pages 5--5, 2005.

\bibitem{YaXiCZ08}
J.~Yang, X.~Zhou, M.~Chrobak, and Y.~Zhang.
\newblock Dynamic thermal management through task scheduling.
\newblock In {\em IEEE International Symposium on Performance Analysis of
  Systems and Software}, 2008.
\newblock To appear.

\end{thebibliography}
